\documentclass[12pt,reqno]{amsart}
\allowdisplaybreaks[4]

\usepackage{enumerate}
\usepackage{CJK}
\usepackage [latin1]{inputenc}
\usepackage{graphicx} 
\usepackage{color}
\usepackage{cite}
\usepackage{subfigure}
\usepackage{graphicx}
\usepackage{epstopdf}
\usepackage{bibentry}
\usepackage{amsthm,amsmath,amssymb}
\usepackage{mathrsfs}
\usepackage{latexsym}

\newtheorem{Definition}{Definition}
\newtheorem{proposition}{Proposition}

\newtheorem{Remark}{Remark}
\newtheorem{theorem}{Theorem}
\newtheorem{Example}{Example}

\begin{document}
	
		\begin{center}
			{\sc \bf On the dispersionless Davey-Stewartson hierarchy: the tau function, the Riemann-Hilbert problem and the Hamilton-Jacobi theory}			
			\vskip 20pt			
			{Ge Yi, Rong Hu, Kelei Tian$^*$ and Ying Xu}			
			\vskip 20pt
			{\it
				School of Mathematics, Hefei University of Technology, Hefei 230601, China
			}

\bigskip
$^*$ Corresponding author:  {\tt kltian@ustc.edu.cn, kltian@hfut.edu.cn}
\bigskip

		\end{center}

\bigskip
\bigskip
\textit{\textbf{Abstract:}} The dDS (dispersionless Davey-Stewartson) hierarchy is constructed by two eigenfunctions of a special Hamiltonian vector field. This hierarchy consists the infinite symmetries of the dDS system. Further, this paper explores the tau function, the Riemann-Hilbert problem and Hamilton-Jacobi theory related to dDS hierarchy.
\bigskip

\textit{\textbf{Keywords:}} dispersionless Davey-Stewartson hierarchy, tau function, Riemann-Hilbert problem, Hamilton-Jacobi theory.
\bigskip
\bigskip

		\section{Introduction}
The DS (Davey-Stewartson) system is one of the most notable (2+1)-dimensional integrable systems. In 1974, a coupled system of nonlinear partial differential equations which describe the evolution of a three-dimensional wave packet in water with a finite depth was derived by Davey and Stewartson using a multi-scale analysis\cite{DavSte}. The general expression of the DS system (parameterized by $\varepsilon >0$) reads as follows
\begin{subequations} \label{DavSte}
\begin{eqnarray}
\textbf{i}\varepsilon q_{t}+\frac{\varepsilon^{2}}{2}(q_{xx}+\sigma^{2} q_{yy})+\delta q \phi=0,
\end{eqnarray}
\begin{eqnarray}
\sigma^{2}\phi_{yy}-\phi_{xx}+(|q|^{2})_{xx}+\sigma^{2}(|q|^{2})_{yy}=0,
\end{eqnarray}
\end{subequations}
in which $x,y,t \in \mathbb{R}$ are the variables of the complex field $q(x,y,t)$ and the real field $\phi(x,y,t)$. We shall refer to (\ref{DavSte}) with $\sigma=1$ the DS-I (Davey-Stewartson-I) system and with $\sigma=\textbf{i}$ the DS-II (Davey-Stewartson-II) system. We also refer to (\ref{DavSte}) with $\delta=1$ the focusing case and  with $\delta=-1$ the defocusing case respectively.
The DS system has undergone thorough study and yielded numerous significant research findings as a typical integrable model.
For DS-II type equations, a productive high-precision numerical method has been proposed in \cite{Klein1}.
It was suggested in \cite{Pogrebkov1}  to use commutator identities on associative algebras as the basis for the method for deriving (2+1)-dimensional nonlinear integrable equations, and this method has now been extended to the standard hierarchy of integrable equations\cite{Pogrebkov1}.
Konopelchenko, Taimanov studied the infinite many symmetries of DS system and pointed out that any symmetry induces an infinite family of geometrically different deformations of tori in $\mathbb{R}^{4}$ preserving the Willmore functional. They defined the DS hierarchy by considering the compatibility of undetermined differential operators in terms of $\partial_{z}$ and $\hat{\partial}_{z}$ \cite{Konopelchenko1,Konopelchenko2,Taimanov} and gave examples of $t_{2}$ and $t_{3}$ flows. Recently, one of the authors  gave a direct expression of the DS hierarchy by two scalar pseudo-differential operators \cite{YiLiao}.

The dispersionless integrable systems are a significant subclass of integrable systems in addition to the classical integrable systems. The dispersionless (semiclassical) limits of the classical integrable systems can sometimes be used to obtain these systems. They are extensively researched and frequently appear in a variety of mathematical physics problems. Jin, Levermore, and McLaughlin have examined the semi-classical limit of the nonlinear Schr\"{o}dinger equation in arbitrary space dimension \cite{Jin1,Jin2}. The semi-classical limit of the nonlinear Schr\"{o}dinger equation in both the defocusing and the focusing cases has been considered by Bronski and McLaughlin in \cite{Bronski1}. Krichever took a deep look into the hierarchy of the dispersionless Lax equations \cite{Krichever1,Krichever2}. Takasaki, Takebe and Guha constructed the solutions of some dispersionless hierarchies and established the dictionary of the twistor geometry and the method of Riemann-Hilbert problem for the dKP hierarchy and the dToda (dispersionless Toda) hierarchy by the twistor theory method\cite{takasaki1,takebe1,takasaki4}.
Gibbons and Kodama studied the solutions of the dKP hierarchy and dToda hierarchy, and developed a more geometric approach by the Hamilton-Jacobi formalism\cite{GibKod,GibKod1,GibKod2,GibKod3}.

The dDS (dispersionless Davey-Stewartson) system and related hierarchy have been first introduced and discussed by Konopelchenko, who applied the dDS hierarchy to describe the highly corrugated surfaces in  $\mathbb{R}^4$ with a slow modulation and propose the quasiclassical generalized Weierstrass representation for such surfaces\cite{konopelchenko2007}. The recent research from one of the authors \cite{yi2018,yi2019} presents that the following dDS system
\begin{subequations} \label{dDS system}
\begin{align}
u_{t}+2(u S_{z})_{z}-2(u S_{\hat{z}})_{\hat{z}}&=0,\\
S_{t}+S^{2}_{z}-S^{2}_{\hat{z}}+\phi&=0,\\
\phi_{z\hat{z}}+2(u_{zz}-u_{\hat{z}\hat{z}})&=0
\end{align}
\end{subequations}
arises from the commutation condition of the Hamiltonian vector fields Lax pair
$$\left[P_{1},P_{2}\right]=0,$$
in which the vector fields operators read as
\begin{subequations} \label{L12}
\begin{align}
P_{1}&=\partial_{\hat{z}}-\{H_{1},\cdot\},\\
P_{2}&=\partial_{t}-\{H_{2},\cdot\},
\end{align}
\end{subequations}
with the Hamiltonians
\begin{subequations}\label{H}
\begin{align}
H_{1}&=S_{\hat{z}}+\frac{u}{p},\label{H1}\\
H_{2}&=(p^{2}-2S_{z} p)+\left[\left(S^{2}_{z}-2 \partial^{-1}_{z} (u_{\hat{z}})\right)+\frac{2u S_{\hat{z}}}{p}+\frac{u^{2}}{p^{2}}\right]. \label{H2}
\end{align}
\end{subequations}
Here and hereafter in this paper, the Poisson bracket is defined as
\begin{equation}
\{A,B\}=\frac{\partial A}{\partial p} \frac{\partial B}{\partial z}-\frac{\partial A}{\partial z}\frac{\partial B}{\partial p}. \nonumber
\end{equation}
This dDS system \eqref{dDS system} is equivalent to the following Zakharov-Shabat equation
$$\frac{\partial H_{1}}{\partial t}-\frac{\partial H_{2}}{\partial \hat{z}}+\{H_{1},H_{2}\}=0.$$
Recently, the dDS system has been studied from the Lie symmetry algebra point of view and some of its exact solutions has been oatained \cite{Gungor}.

The system of integrable partial differential equations is usually related to the hierarchy of partial differential equations that define an infinite number of symmetries. According to \eqref{H}, since the dDS case includes both the positive and negative parts, both the $0$ and $\infty$ are the singular points in the complex $p$-plane. Guha, Takasaki and Takebe have made a series of important and insightful studies for dispersionless integrable hierarchies \cite{takebe1,takasaki1,takasaki2,takasaki3,takasaki4,takasaki5}. We will study the tau function, the Riemann-Hilbert problem and Hamilton-Jacobi theory for the dDS hierarchy based on the previous research \cite{yi2018,yi2019}.

This paper is organized as follows. In section 2, we introduce the dDS hierarchy which includes infinitely many compatible flows, derive the twistor structure and Lax-Sato formlism for the dDS hierarchy, and show some meaningful examples. In section 3, we present the existence for the tau function of the dDS hierarchies. In section 4, we discuss the related Riemann-Hilbert problem. In section 5, we study the Hamilton-Jacobi theory. In section 6, we summarise the results of this paper.

\bigskip
\section{The dispersionless Davey-Stewartson hierarchy}

A new hierarchy of compatible partial differential equations was introduced by one of the authors in \cite{yi2019}. This hierarchy  consists the infinite symmetries of
the dDS system \eqref{dDS system} is called as the dDS hierarchy. In this section, we will review some basic concepts and properties of this dispersionless hierarchy.

As an extension of the dDS system, the construction of the dDS hierarchy is based on two eigenfunctions of this following special Hamiltonian vector field
\begin{eqnarray}
P&=&\partial_{\hat{z}}-\{\hat{H},\cdot\}\nonumber\\
&=&\partial_{\hat{z}}-\{v+\frac{u}{p},\cdot\}  \nonumber\\
&=& \partial_{\hat{z}}+\frac{u}{p^2}\partial_z+(v_z+\frac{u_z}{p})\partial_p,\nonumber
\end{eqnarray}
where $p$ is a complex parameter and  $u=u(t),v=v(t)$ depend on complex variables $t=(t_{mn})$ ($m,n\in \mathbb{N}$, $m+n \geq 1$, and we denote $t_{10} \equiv z,t_{01} \equiv \hat{z}$ in this paper).

Given an arbitrary closed curve $\Gamma$ around the origin of the complex $p$-plane, the two eigenfunctions $\mathcal{L}$ and $\hat{\mathcal{L}}$ of $P$ are described by the two formal Laurent series
\begin{subequations}\label{Ls}
\begin{align}
\mathcal{L}&=p+\sum_{i\le0}f_i(t)p^i,\label{L}\\
\hat{\mathcal{L}}&=\frac{u(t)}{p}+\sum_{i\ge0}g_i(t)p^i.\label{Lhat}
\end{align}
\end{subequations}

\begin{Remark}\label{Remark1}
Here $f_{i}$ and $g_{i}$ both depend on the two real functions $u,v$ and their derivatives or integrals with respect to the independent variables $z,\hat{z}$. Then the coefficients $f^{(m)}_{i}$ and $g^{(n)}_{i}$ of the Laurent expansions
\begin{align}
\mathcal{L}^{m}&=p^{m}+\sum_{i \leq m-1}f^{(m)}_{i} p^{i},\nonumber\\
\hat{\mathcal{L}}^{n}&=\frac{u^{n}}{p^{n}}+\sum_{i \geq 1-n}g^{(n)}_{i} p^{i},\nonumber
\end{align}
also depend on $u,v$ and their derivatives or integrals with respect to the independent variables $z,\hat{z}$.
\end{Remark}
The second Hamiltonian associated with the dDS system (\ref{dDS system}) can be expressed as
$(\mathcal{L}^{2})_{>0}+(\hat{\mathcal{L}}^{2})_{\leq 0}$. By introducing the following Hamiltonian, this concept can be extended to any non-negative integer $m,n$ as
\begin{eqnarray} \label{Hmn}
&&H_{mn}=(\mathcal{L}^{m})_{>0}+(\hat{\mathcal{L}}^{n})_{\leq 0},\nonumber \\
&&H_{m0}=(\mathcal{L}^{m})_{>0},~~H_{0n}=(\hat{\mathcal{L}}^{n})_{\leq 0}.\nonumber
\end{eqnarray}
Here and hereafter in this paper, the symbol $()_{>0}$ stands for extracting the positive powers of $p$ and similarly the symbol $()_{\leq 0}$ stands for extracting the non-positive part.

Then the Zakharov-Shabat equations
\begin{eqnarray} \label{simplezerocurvature}
\frac{\partial \hat{H}}{\partial t_{mn}}-\frac{\partial H_{mn}}{\partial \hat{z}}+\{\hat{H},H_{mn}\}=0
\end{eqnarray}
are equivalent to the following $(2+1)$-dimensional closed systems for unknowns $u$ and $v$,
\begin{subequations} \label{simplesystems}
\begin{eqnarray}
u_{t_{mn}}-(u f^{(m)}_{1})_{z}-g^{(n)}_{-1,\hat{z}}=0,
\end{eqnarray}
\begin{eqnarray}
v_{t_{mn}}+f^{(m)}_{0,\hat{z}}-g^{(n)}_{0,\hat{z}}=0.
\end{eqnarray}
\end{subequations}

The concept of dDS hierarchy is a direct extension of the above special case, which includes a specific Hamiltonian $\hat{H}$.

\begin{Definition}\label{Zerocurvature}
The dDS hierarchy is defined by the Zakharov-Shabat equations
\begin{eqnarray} \label{zerocurvature}
\frac{\partial H_{mn}}{\partial t_{kl}}-\frac{\partial H_{kl}}{\partial t_{mn}}+\{H_{mn},H_{kl}\}=0.
\end{eqnarray}
\end{Definition}


The Zakharov-Shabat equation \eqref{zerocurvature} occurs in the analysis of the  self-dual vacuum Einstein and hyper-K\"{a}hler geometry \cite{takasaki6}. In the study of the vacuum Einstein equation's hyper-K\"{a}hler version, a K\"{a}hler-like 2-form and associated ``Darboux coordinates" play a central role \cite{gindikin,hitchin}.

By introducing an exterior differential 2-form
\begin{equation} \label{omega}
\omega=\sum_{m+n \geq 1} dH_{mn}\wedge dt_{mn}=d p \wedge d z+d \hat{H} \wedge d \hat{z}+\sum_{m+n \geq 2} dH_{mn}\wedge dt_{mn},
\end{equation}
we can obtain the another expression of the dDS hierarchy which resembles the twistor structure.

This 2-form is a degenerate symplectic form because $\omega$ is closed
$$d\omega=0,$$
and the Zakharov-Shabat equations \eqref{zerocurvature} can be cast into a compact form as
$$\omega\wedge\omega=0.$$
These two relations mean that there are two functions $Q$ and $R$, which give a pair of Darbox coordinates
$$\omega=dQ\wedge dR.$$

\begin{Definition}\label{omegadDS}
The dDS hierarchy is equivalent to
\begin{equation}\label{dDS2}
\omega=d\mathcal{L}\wedge d\mathcal{M}=d\hat{\mathcal{L}}\wedge d\hat{\mathcal{M}},
\end{equation}
in which
\begin{subequations} \label{eigen}
\begin{align}
\mathcal{M}&=\sum_{m+n\ge 2}mt_{mn}\mathcal{L}^{m-1}+z+\sum_{i=1}^{\infty}v_{i+1}\mathcal{L}^{-i-1},\label{M}\\
\hat{\mathcal{M}}&=\sum_{m+n\ge 2}nt_{mn}\hat{\mathcal{L}}^{n-1}+\hat{z}+\sum_{i=1}^{\infty}\hat{v}_{i+1}\hat{\mathcal{L}}^{-i-1}.\label{Mhat}
\end{align}
\end{subequations}
The symbol $``d"$  stands for total differentiation in $p,z,\hat{z}$ and $t_{mn}~(m+n \geq 2)$.
\end{Definition}
By considering the definition (\ref{omega}) and the relation (\ref{omegadDS}), one obtains the following third equivalent definition of the dDS hierarchy.

\begin{Definition}\label{Lax}
 The dDS hierarchy consists of the following Lax-Sato equations
\begin{eqnarray} \label{LaxS}
\frac{\partial K}{\partial t_{mn}}=\{H_{mn},K\}
\end{eqnarray}
in which eigenfunctions $K=\mathcal{L},\mathcal{M},\hat{\mathcal{L}},\hat{\mathcal{M}}$ satisfy the following canonical Poisson relations
$$\{\mathcal{L},\mathcal{M}\}=1,\ \ \ \{\hat{\mathcal{L}},\hat{\mathcal{M}}\}=1.$$
\end{Definition}

To the end of this section, we present some meaningful examples of the dDS hierarchy below.
\begin{Example}
By taking $m=0,n=1;k=2,l=2$ and $t_{01}=\hat{z},t_{22}=t$, Hamiltonians read as follows
\begin{subequations}
\begin{eqnarray}
H_{01}= v+\frac{u}{p},
\end{eqnarray}
\begin{eqnarray}
H_{22}&=&(p^{2}+2f_{0} p)+ [(v^{2}+2u g_{1})+\frac{2uv}{p}+\frac{u^{2}}{p^{2}}] \nonumber\\
&\equiv& (p^{2}+2 f p)+ [w+\frac{2uv}{p}+\frac{u^{2}}{p^{2}}].
\end{eqnarray}
\end{subequations}
Then nonlinear system arise from the Zakharov-Shabat equation  (\ref{zerocurvature}) read as
\begin{subequations}
\begin{eqnarray}
2u_{\hat{z}}+w_{z}-2vv_{z}=0,
\end{eqnarray}
\begin{eqnarray}
u_{t}-2(uf)_{z}-2(uv)_{\hat{z}}=0,
\end{eqnarray}
\begin{eqnarray}
v_{t}-2u_{z}-2fv_{z}- w_{\hat{z}}=0,
\end{eqnarray}
\begin{eqnarray}
f_{\hat{z}}+v_{z}=0.
\end{eqnarray}
\end{subequations}
which can be simplified exactly to the  dDS system (\ref{dDS system}) with the choice $v=S_{\hat{z}}$.

\end{Example}
\begin{Example}
With equations \eqref{simplesystems} and \eqref{simplezerocurvature}, choosing the one of the two  Hamiltonian as $\hat{H}=H_{01}$. The integrable flow equation can be obtained simply by \eqref{simplesystems}. Taking another Hamiltonian as $H_{33}$ and letting $v=S_{\hat{z}}$,
the functions $f^{(3)}_{0},f^{(3)}_{1},g^{(3)}_{0}$ and $g^{(3)}_{-1}$ in the expression \eqref{simplesystems} are
\begin{align}
f^{(3)}_{0}&=-S^{3}_{z}+6S_{z} V+\partial^{-1}_{\hat{z}} (u S_{zz}-S_{z\hat{z}} V),\ \ \ f^{(3)}_{1}=3(S^{2}_{z}-V),\nonumber\\
g^{(3)}_{0}&=S^{3}_{\hat{z}}-6S_{\hat{z}} W+\partial^{-1}_{z} (S_{z\hat{z}} W-u S_{\hat{z}\hat{z}}),\ \ \ g^{(3)}_{-1}=3u(S^{2}_{\hat{z}}-W),\nonumber
\end{align}
in which
\begin{eqnarray}\nonumber
W_{z}=u_{\hat{z}},~~~~~~V_{\hat{z}}=u_{z}.
\end{eqnarray}
Then the following system is obtained,
\begin{align}
&u_{t_{33}}-3\left[u(S^{2}_{z}-V)\right]_{z}-3\left[u(S^{2}_{\hat{z}}-W)\right]_{\hat{z}}=0,\nonumber\\
&S_{t_{33}}-(S^{3}_{z}+S^{3}_{\hat{z}})+3(S_{z} V+S_{\hat{z}} W)+3\phi=0,\nonumber\\
&\phi_{z\hat{z}}=(uS_{z})_{zz}+(uS_{\hat{z}})_{\hat{z}\hat{z}},\nonumber
\end{align}
which is analogous to the dDS system (\ref{dDS system}).
\end{Example}

\begin{Remark}
Similar to the Manakov-Santini hierarchy \cite{manakov2,manakov3}, the general unreduced hierarchy without standardized Poisson relations can also be defined by the following Lax-Sato equations
\begin{eqnarray}\nonumber
\mathcal{P}_{mn} (K)=\frac{\partial K}{\partial t_{mn}}-F_{mn} \frac{\partial K}{\partial z}+G_{mn} \frac{\partial K}{\partial p}=0
\end{eqnarray}
for eigenfunctions $K=\mathcal{L},\mathcal{M},\hat{\mathcal{L}},\hat{\mathcal{M}}$ defined in (\ref{eigen}), in which
\begin{align}
&F_{mn}=(m  \mathcal{L}^{m-1} \mathcal{L}_{p} J^{-1})_{\geq 0}+(n \hat{\mathcal{L}}^{n-1} \hat{\mathcal{L}}_{p}\hat{J}^{-1})_{< 0},\nonumber\\
&G_{mn}=(m \mathcal{L}^{m-1} \mathcal{L}_{z} J^{-1})_{>0}+(n \hat{\mathcal{L}}^{n-1} \hat{\mathcal{L}}_{z}\hat{J}^{-1})_{\leq 0},\nonumber
\end{align}
here $J=\{\mathcal{L},\mathcal{M}\},\hat{J}=\{ \hat{\mathcal{L}}, \hat{\mathcal{M}}\}$ and $ J^{-1},\hat{J}^{-1}$ are their formal inverses respectively.
Obviously, this definition is equivalent to
\begin{eqnarray}
[\mathcal{P}_{mn},\mathcal{P}_{kl}]=0.\nonumber
\end{eqnarray}
\end{Remark}

\bigskip
		\section{The tau function}
The tau function encompasses all aspects of integrable systems that can be used to define  solutions for  the whole integrable hierarchy. By applying the definitions of dDS hierarchy introduced in the preceding section, we can discuss the existence of the tau function from the 2-form $\eqref{omega}$.
\begin{theorem}\label{tauT}
There exists the tau function $\tau_{dDS}(t)$ satisfying
\begin{equation}\label{tau}
d \log\tau_{dDS}(t)=\sum_{m+n\ge2}(v_{m+1}+\hat{v}_{n+1})dt_{mn}.
\end{equation}

\end{theorem}
This is due to the basic fact that the right hand side of the equation $(\ref{tau})$ is a closed form and we will prove this below.

Firstly, let us introduce the notion of formal residue of 1-forms
$$res\sum a_np^ndp=a_{-1}.$$

For the eigenfunctions $\mathcal{L}$ and $\hat{\mathcal{L}}$, similar to the literature \cite{takasaki1}, the following formulas are also true
\begin{align}
&res\ \mathcal{L}^nd\mathcal{L}=\delta_{n,-1},\nonumber\\
&res\ \hat{\mathcal{L}}^nd\hat{\mathcal{L}}=-\delta_{n,-1},\ \ \ n\in \mathbb{Z}.\nonumber
\end{align}

\begin{proposition}\label{v}
The coefficients $v_i$ and $\hat{v}_i$ from the Orlov eigenfunctions $\mathcal{M}$ and $\hat{\mathcal{M}}$, which are defined by (\ref{eigen}), satisfy
\begin{subequations}
\begin{align}
\frac{\partial \hat{v}_{i+1}}{\partial t_{mn}}=-res\ \hat{\mathcal{L}}^{i}dH_{mn}.\label{dv2}\\
\frac{\partial v_{i+1}}{\partial t_{mn}}=res\ \mathcal{L}^{i}dH_{mn},\label{v1}
\end{align}
\end{subequations}
\end{proposition}
\begin{proof}
We only give the proof of the equations $(\ref{dv2})$. By the chain rule of differentiation, one obtains
\begin{equation}\label{v2}
\frac{\partial\hat{\mathcal{M}}}{\partial t_{mn}}=\left.\frac{\partial\hat{\mathcal{M}}}{\partial \hat{\mathcal{L}}}\right|_{t,\hat{v}\ fixed}\frac{\partial\hat{\mathcal{L}}}{\partial t_{mn}}+n\hat{\mathcal{L}}^{n-1}+\sum_{i=1}^{\infty}\frac{\partial\hat{v}_{i+1}}{\partial t_{mn}}\hat{\mathcal{L}}^{-i-1}.\nonumber
\end{equation}
Therefore, using the equations of the eigenfunctions $\mathcal{L}$ and $\hat{\mathcal{L}}$ as well, one has
\begin{align}
 \frac{\partial \hat{v}_{i+1}}{\partial t_{mn}}
 &=-res\ \hat{\mathcal{L}}^i\left(\frac{\partial\hat{\mathcal{M}}}{\partial t_{mn}}-\left.\frac{\partial\hat{\mathcal{M}}}{\partial \hat{\mathcal{L}}}\right|_{t,\hat{v}\ fixed}\frac{\partial\hat{\mathcal{L}}}{\partial t_{mn}}\right) d\hat{\mathcal{L}}\nonumber\\
& =-res\ \hat{\mathcal{L}}^i\left(\frac{\partial\hat{\mathcal{M}}}{\partial t_{mn}}\right) d\hat{\mathcal{L}}+res\ \hat{\mathcal{L}}^i\left(\frac{\partial\hat{\mathcal{L}}}{\partial t_{mn}}\right) d\hat{\mathcal{M}}\nonumber\\
 &=-res\ \hat{\mathcal{L}}^i\left\{H_{mn},\hat{\mathcal{M}}\right\} d\hat{\mathcal{L}}+res\  \hat{\mathcal{L}}^i\left\{H_{mn},\hat{\mathcal{L}}\right\} d\hat{\mathcal{M}}\nonumber\\
 &=-res\ \hat{\mathcal{L}}^i\left(\frac{\partial H_{mn}}{\partial p}\left(\frac{\partial\hat{\mathcal{M}}}{\partial z}\frac{\partial\hat{\mathcal{L}}}{\partial p}-\frac{\partial\hat{\mathcal{M}}}{\partial p}\frac{\partial\hat{\mathcal{L}}}{\partial z}\right)+\frac{\partial H_{mn}}{\partial z}\left(\frac{\partial\hat{\mathcal{M}}}{\partial p}\frac{\partial\hat{\mathcal{L}}}{\partial p}-\frac{\partial\hat{\mathcal{M}}}{\partial p}\frac{\partial\hat{\mathcal{L}}}{\partial p}\right)\right)dp\nonumber\\
 &=-res\ \hat{\mathcal{L}}^idH_{mn},\nonumber
\end{align}
where we have used the Lax equations $(\ref{LaxS})$ and the canonical Possion relation $\{\hat{\mathcal{L}},\hat{\mathcal{M}}\}=1$. The equation $(\ref{v1})$ can be proved similarly.
\end{proof}

\begin{proposition}\label{tauv}
The coefficients $v_{m+1}$ and $\hat{v}_{n+1}$ from the Orlov eigenfunctions $\mathcal{M}$ and $\hat{\mathcal{M}}$,which are defined by (\ref{eigen}), satisfy the condition
\begin{equation}
\frac{\partial(v_{m+1}+\hat{v}_{n+1})}{\partial t_{kl}}=\frac{\partial(v_{k+1}+\hat{v}_{l+1})}{\partial t_{mn}},\ \ \ m,n\in \mathbb{Z}_+.
\end{equation}
\end{proposition}
\begin{proof}
By applying the  Proposition 1, one obtanis
\begin{align}
&\frac{\partial(v_{m+1}+\hat{v}_{n+1})}{\partial t_{kl}}-\frac{\partial(v_{k+1}+\hat{v}_{l+1})}{\partial t_{mn}}\nonumber\\
&=res(\mathcal{L}^m-\hat{\mathcal{L}}^n)dH_{kl}-res(\mathcal{L}^k-\hat{\mathcal{L}}^l)dH_{mn}\nonumber\nonumber\\
&=res((\mathcal{L}^m)_{\ge0}-\hat{(\mathcal{L}}^n)_{\ge0})d(\hat{\mathcal{L}}^l)_{\le-1}+res((\mathcal{L}^m)_{\le-1}-(\hat{\mathcal{L}}^n)_{\le-1})d(\mathcal{L}^k)_{>0}\nonumber\\
&\ \ -res((\mathcal{L}^k)_{\ge0}-(\hat{\mathcal{L}}^l)_{\ge0})d(\hat{\mathcal{L}}^n)_{\le-1}-res((\mathcal{L}^k)_{\le-1}-(\hat{\mathcal{L}})_{\le-1}^l)d(\mathcal{L}^m)_{>0}\nonumber\\
&=res\hat{\mathcal{L}}^l d\hat{\mathcal{L}}^n+res\mathcal{L}^m d\mathcal{L}^k,\nonumber
\end{align}
which vanishes because $m$ and $n$ are positive integers.
\end{proof}
Based on the above facts, we arrive the conclusion that the right hand side of $(\ref{tau})$ is closed. The Theorem 1 shows that the coefficients $u_i$, $\hat{u}_i$, $v_i$ and $\hat{v}_i$ from the two pairs of eigenfunctions $(\mathcal{L},\mathcal{M})$ and $(\hat{\mathcal{L}},\hat{\mathcal{M}})$  can be expressed by this tau function $\tau_{dDS}$. For the case of $n=0$, it reduces exactly to the tau function $\tau_{dKP}$ of dKP hierarchy\cite{takasaki1}.

\bigskip
\section{The Riemann-Hilbert problem}

The study of the twistor construction for the self-dual vacuum Einstein equation and its hyper-K\"{a}hler version \cite{Penrose,Hitchin} is important in geometry and mathematical physics. This theory can be also extended to consider the dDS hierarchy and the key step is to solve the following
\begin{subequations}\label{RHP}
\begin{align}
f(\mathcal{L},\mathcal{M})&=\hat{f}(\hat{\mathcal{L}},\hat{\mathcal{M}}),\\ g(\mathcal{L},\mathcal{M})&=\hat{g}(\hat{\mathcal{L}},\hat{\mathcal{M}}),
\end{align}
\end{subequations}
where $(f,g)$ and $(\hat{f}, \hat{g})$ are two pairs of  holomorphic functions  satisfying the the canonical Poisson relation $\{f,g\}=\{\hat{f},\hat{g}\}=1$.  In this section, we will describe the details of the Riemann-Hilbert problem based on the two pairs of eigenfunctions $(\mathcal{L},\mathcal{M}),\  (\hat{\mathcal{L}},\hat{\mathcal{M}})$ defined in $(\ref{Ls})$ $(\ref{eigen})$.

\begin{theorem}\label{RHPp}
The solutions $(\mathcal{L},\mathcal{M})$ and $(\hat{\mathcal{L}},\hat{\mathcal{M}})$ of the Riemann-Hilbert problem $(\ref{RHP})$ solve the dDS hierarchy. Namely, they satisfy the Lax-Sato equations $(\ref{LaxS})$  and the canonical Poisson relations $\{\mathcal{L},\mathcal{M}\}=\{\hat{\mathcal{L}},\hat{\mathcal{M}}\}=1$. The holomorphic functions  $(f,g,\hat{f},\hat{g})$  are defined as the twistor data of this solution.
\end{theorem}
\begin{proof}
Firstly
\begin{equation}\label{proof}
\left(
\begin{matrix}
\dfrac{\partial f}{\partial \mathcal{L}} &
\dfrac{\partial f}{\partial \mathcal{M}}\\
\dfrac{\partial g}{\partial \mathcal{L}} &
\dfrac{\partial g}{\partial \mathcal{M}}
\end{matrix}
\right)
\left(
\begin{matrix}
\dfrac{\partial \mathcal{L}}{\partial p} & \dfrac{\partial \mathcal{L}}{\partial z}\\
\dfrac{\partial \mathcal{M}}{\partial p} & \dfrac{\partial \mathcal{M}}{\partial z}
\end{matrix}
\right)
=
\left(
\begin{matrix}
\dfrac{\partial \hat{f}}{\partial \hat{\mathcal{L}}}& \dfrac{\partial \hat{f}}{\partial \hat{\mathcal{M}}}\\
\dfrac{\partial \hat{g}}{\partial \hat{\mathcal{L}}} & \dfrac{\partial \hat{g}}{\partial \hat{\mathcal{M}}}
\end{matrix}
\right)
\left(
\begin{matrix}
\dfrac{\partial \hat{\mathcal{L}}}{\partial p} & \dfrac{\partial \hat{\mathcal{L}}}{\partial z}\\
\dfrac{\partial \hat{\mathcal{M}}}{\partial p} & \dfrac{\partial \hat{\mathcal{M}}}{\partial z}
\end{matrix}
\right).
\end{equation}
By considering the determinant of the both hand sides of this equation and using the relations $$\{f,g\}=\{\hat{f},\hat{g}\}=1, $$
one abstains
$$\{\mathcal{L},\mathcal{M}\}=\{\hat{\mathcal{L}},\hat{\mathcal{M}}\}.$$
Then
\begin{align}
\{\mathcal{L},\mathcal{M}\}&=\frac{\partial \mathcal{L}}{\partial p}\frac{\partial \mathcal{M}}{\partial z}-\frac{\partial \mathcal{L}}{\partial z}\frac{\partial \mathcal{M}}{\partial p}\nonumber\\
&=\frac{\partial \mathcal{L}}{\partial p}\left(\left.\frac{\partial \mathcal{M}}{\partial \mathcal{L}}\right|_{t,v \ fixed}\frac{\partial \mathcal{L}}{\partial z}+1+\sum_{i=1}^{\infty}\frac{\partial v_{i+1}}{\partial z}\mathcal{L}^{-i-1}\right)-\frac{\partial \mathcal{L}}{\partial z}\left.\frac{\partial \mathcal{M}}{\partial \mathcal{L}}\right|_{t,v\ fixed}\frac{\partial\mathcal{L}}{\partial p}\nonumber.
\end{align}
The formal Laurent expansion of $\{\hat{\mathcal{L}},\hat{\mathcal{M}}\}$ contains only nonnegative power of $p$, therefore $$\{\mathcal{L},\mathcal{M}\}=\{\hat{\mathcal{L}},\hat{\mathcal{M}}\}=1.$$
This gives the canonical Poisson relations for the solutions of the equations $(\ref{RHP})$.
Further, differentiating $(\ref{RHP})$ by $t_{mn}$ gives
\begin{equation}\label{proof2}
\left(
\begin{matrix}
\dfrac{\partial f(\mathcal{L},\mathcal{M})}{\partial \mathcal{L}} &
\dfrac{\partial f(\mathcal{L},\mathcal{M})}{\partial \mathcal{M}}\\
\dfrac{\partial g(\mathcal{L},\mathcal{M})}{\partial \mathcal{L}} &
\dfrac{\partial g(\mathcal{L},\mathcal{M})}{\partial \mathcal{M}}
\end{matrix}
\right)
\left(
\begin{matrix}
\dfrac{\partial \mathcal{L}}{\partial t_{mn}} \\
\dfrac{\partial \mathcal{M}}{\partial t_{mn}}
\end{matrix}
\right)
=
\left(
\begin{matrix}
\dfrac{\partial \hat{f}(\hat{\mathcal{L}},\hat{\mathcal{M}})}{\partial \hat{\mathcal{L}}}& \dfrac{\partial \hat{f}(\hat{\mathcal{L}},\hat{\mathcal{M}})}{\partial \hat{\mathcal{M}}}\\
\dfrac{\partial \hat{g}(\hat{\mathcal{L}},\hat{\mathcal{M}})}{\partial \hat{\mathcal{L}}} & \dfrac{\partial \hat{g}(\hat{\mathcal{L}},\hat{\mathcal{M}})}{\partial \hat{\mathcal{M}}}
\end{matrix}
\right)
\left(
\begin{matrix}
\dfrac{\partial \hat{\mathcal{L}}}{\partial t_{mn}}\\
\dfrac{\partial \hat{\mathcal{M}}}{\partial t_{mn}}
\end{matrix}
\right).
\end{equation}
By combining $(\ref{proof})$, we can rewrite the formula $(\ref{proof2})$ as follows
\begin{equation}\label{proof3}
\left(
\begin{matrix}
\dfrac{\partial \mathcal{M}}{\partial z} &
-\dfrac{\partial\mathcal{L}}{\partial z}\\
-\dfrac{\partial\mathcal{M}}{\partial p} &
\dfrac{\partial \mathcal{L}}{\partial p}
\end{matrix}
\right)
\left(
\begin{matrix}
\dfrac{\partial \mathcal{L}}{\partial t_{mn}} \\
\dfrac{\partial \mathcal{M}}{\partial t_{mn}}
\end{matrix}
\right)
=
\left(
\begin{matrix}
\dfrac{\partial \hat{\mathcal{M}}}{\partial z}& -\dfrac{\partial \hat{\mathcal{L}}}{\partial z}\\
-\dfrac{\partial\hat{\mathcal{M}}}{\partial p}& \dfrac{\partial \hat{\mathcal{L}}}{\partial p}
\end{matrix}
\right)
\left(
\begin{matrix}
\dfrac{\partial \hat{\mathcal{L}}}{\partial t_{mn}}\\
\dfrac{\partial \hat{\mathcal{M}}}{\partial t_{mn}}
\end{matrix}
\right).
\end{equation}
The first component of the left hand side for $(\ref{proof3})$ is given by
\begin{align}
&\ \ \dfrac{\partial \mathcal{M}}{\partial z}\dfrac{\partial \mathcal{L}}{\partial t_{mn}}-\dfrac{\partial \mathcal{L}}{\partial z}\dfrac{\partial \mathcal{M}}{\partial t_{mn}}\nonumber\\
&=\left(\left.\dfrac{\partial \mathcal{M}}{\partial \mathcal{L}}\right|_{p,v\ fixed}\dfrac{\partial \mathcal{L}}{\partial z}+1+\sum_{j=1}^{\infty}\dfrac{\partial v_{j+1}}{\partial z}\mathcal{L}^{-j-1}\right)\dfrac{\partial \mathcal{L}}{\partial t_{mn}}\nonumber\\
&\ \ -\left(\left.\dfrac{\partial \mathcal{M}}{\partial \mathcal{L}}\right|_{p,v\ fixed}\dfrac{\partial \mathcal{L}}{\partial t_{mn}}+m\mathcal{L}^{m-1}+\sum_{j=1}^{\infty}\dfrac{\partial v_{j+1}}{\partial t_{mn}}\mathcal{L}^{-j-1}\right)\dfrac{\partial \mathcal{L}}{\partial z}\nonumber\\
&=-\dfrac{\partial \mathcal{L}^m}{\partial z}+\dfrac{\partial \mathcal{L}}{\partial t_{mn}}+\left(\sum_{j=1}^{\infty}\dfrac{\partial v_{j+1}}{\partial t_{mn}}\mathcal{L}^{-j-1}\right)\dfrac{\partial \mathcal{L}}{\partial t_{mn}}-\left(\sum_{j=1}^{\infty}\dfrac{\partial v_{j+1}}{\partial t_{mn}}\mathcal{L}^{-j-1}\right)\dfrac{\partial \mathcal{L}}{\partial z},\nonumber
\end{align}
which contains limited positive powers of $p$. Similarly, the second component of the right hand side of $(\ref{proof3})$ is given by
\begin{align}
&\frac{\partial \hat{\mathcal{M}}}{\partial z}\frac{\partial \hat{\mathcal{L}}}{\partial t_{mn}}-\frac{\partial \hat{\mathcal{L}}}{\partial z}\frac{\partial \hat{\mathcal{M}}}{\partial t_{mn}}\nonumber\\
&=\left(\left.\frac{\partial \hat{\mathcal{M}}}{\partial \hat{\mathcal{L}}}\right|_{p,\hat{v}\ fixed}\frac{\partial \hat{\mathcal{L}}}{\partial z}+\sum_{j=1}^{\infty}\frac{\partial \hat{v}_{j+1}}{\partial z}\hat{\mathcal{L}}^{-j-1}\right)\frac{\partial \hat{\mathcal{L}}}{\partial t_{mn}}\nonumber\\
&\ \ -\left(\left.\frac{\partial \hat{\mathcal{M}}}{\partial \hat{\mathcal{L}}}\right|_{p,\hat{v}\ fixed}\frac{\partial \hat{\mathcal{L}}}{\partial t_{mn}}+n\hat{\mathcal{L}}^{n-1}+\sum_{j=1}^{\infty}\frac{\partial \hat{v}_{j+1}}{\partial t_{mn}}\hat{\mathcal{L}}^{-j-1}\right)\frac{\partial \hat{\mathcal{L}}}{\partial z}\nonumber\\
&=-\frac{\partial \hat{\mathcal{L}}^n}{\partial z}+\left(\sum_{j=1}^{\infty}\frac{\partial \hat{v}_{j+1}}{\partial t_{mn}}\hat{\mathcal{L}}^{-j-1}\right)\frac{\partial \hat{\mathcal{L}}}{\partial t_{mn}}-\left(\sum_{j=1}^{\infty}\frac{\partial \hat{v}_{j+1}}{\partial t_{mn}}\hat{\mathcal{L}}^{-j-1}\right)\frac{\partial \hat{\mathcal{L}}}{\partial z},\nonumber
\end{align}
which contains limited non-positive powers of $p$. By considering both sides of $(\ref{proof3})$, one obtains
$$
\frac{\partial \mathcal{M}}{\partial z}\frac{\partial \mathcal{L}}{\partial t_{mn}}-\frac{\partial \mathcal{L}}{\partial z}\frac{\partial \mathcal{M}}{\partial t_{mn}}=\frac{\partial \hat{\mathcal{M}}}{\partial z}\frac{\partial \hat{\mathcal{L}}}{\partial t_{mn}}-\frac{\partial \hat{\mathcal{L}}}{\partial z}\frac{\partial \hat{\mathcal{M}}}{\partial t_{mn}}=-\frac{\partial (\hat{\mathcal{L}}^n)_{\le -1}}{\partial z}-\frac{\partial (\mathcal{L}^m)_{> 0}}{\partial z},
$$
$$
\frac{\partial \mathcal{M}}{\partial p}\frac{\partial \mathcal{L}}{\partial t_{mn}}-\frac{\partial \mathcal{L}}{\partial p}\frac{\partial \mathcal{M}}{\partial t_{mn}}=\frac{\partial \hat{\mathcal{M}}}{\partial p}\frac{\partial \hat{\mathcal{L}}}{\partial t_{mn}}-\frac{\partial \hat{\mathcal{L}}}{\partial p}\frac{\partial \hat{\mathcal{M}}}{\partial t_{mn}}=-\frac{\partial (\hat{\mathcal{L}}^n)_{\le -1}}{\partial p}-\frac{\partial (\mathcal{L}^m)_{> 0}}{\partial p}.
$$
Then these equations can be readily solved as
\begin{align}
\frac{\partial \mathcal{L}}{\partial t_{mn}}&=\{H_{mn},\mathcal{L}\},\
 \ \ \frac{\partial \hat{\mathcal{L}}}{\partial t_{mn}}=\{H_{mn},\hat{\mathcal{L}}\},\nonumber\\
\frac{\partial \mathcal{M}}{\partial t_{mn}}&=\{H_{mn},\mathcal{M}\},\ \ \ \frac{\partial \hat{\mathcal{M}}}{\partial t_{mn}}=\{H_{mn},\hat{\mathcal{M}}\}.\nonumber
\end{align}
\end{proof}

\bigskip
\section{The Hamilton-Jacobi theory}
It is well known that, in classical mechanics, the Hamilton-Jacobi theory is a way to integrate a Hamiltonian system through an appropriate canonical transformation. The generating function of this transformation, which provides an integration of the original system, is a partial differential equation. The Hamilton-Jacobi theory derives from the analytical mechanics and has been an important method in classical mechanics, theoretical physics, differential equations and differential geometry. In this section, we will continue to discuss the Hamilton-Jacobi theory  related to the dDS hierarchy.

Firstly, we define the functions $p(\lambda,\mu,t)=p$, $z(\lambda,\mu,t)=z$ implicitly by
$$\mathcal{L}(p,z,t)=\lambda,\ \ \ \mathcal{M}(p,z,t)=\mu,$$ with the parameters $\lambda$, $\mu$.

Similarly, we can define functions
 $\hat{p}(\hat{\lambda},\hat{\mu},t)=p$, $\hat{z}(\hat{\lambda},\hat{\mu},t)=z$ implicitly
by
$$ \hat{\mathcal{L}}(p,z,t)=\hat{\lambda},\ \ \ \hat{\mathcal{M}}(p,z,t)=\hat{\mu},$$ with the parameters $\hat{\lambda}$, $\hat{\mu}$.

Then we can define  a couple of Hamiltonian systems as described in the following theorem.
\begin{theorem}
Both $(p(\lambda,\mu,t),\ z(\lambda,\mu,t))$ and $(\hat{p}(\hat{\lambda},\hat{\mu},t),\ \hat{z}(\hat{\lambda},\hat{\mu},t))$ satisfy the multi-time Hamiltonian system
\begin{equation}\nonumber
\left\{
\begin{aligned}
\frac{dp}{dt_{mn}}&=\frac{\partial H_{mn}}{\partial z}\\
\frac{dz}{dt_{mn}}&=-\frac{\partial H_{mn}}{\partial p}
\end{aligned}
\right.\ \ ,\ \ \ \
\left\{
\begin{aligned}
\frac{d\hat{p}}{dt_{mn}}&=\frac{\partial H_{mn}}{\partial \hat{z}}\\ \frac{d\hat{z}}{dt_{mn}}&=-\frac{\partial H_{mn}}{\partial \hat{p}}
\end{aligned}
\right.\ \ ,
\end{equation}
with time-dependent Hamiltonians $H_{mn}$.
\end{theorem}
\begin{proof}
In this proof, we abbreviate $(p(\lambda,\mu,t), z(\lambda,\mu,t))$,  $(\hat{p}(\hat{\lambda},\hat{\mu},t), \hat{z}(\hat{\lambda},\hat{\mu},t))$ as $(p(t),z(t))$ , $(\hat{p}(t),\hat{z}(t))$ to avoid complicated symbols.

According to the definition, the following identities are true
\begin{align}
\mathcal{L}(p(t),z(t),t)=\lambda,\ \ \ & \mathcal{M}(p(t),z(t),t)=\mu,\nonumber\\
\hat{\mathcal{L}}(\hat{p}(t),\hat{z}(t),t)=\hat{\lambda},\ \ \ & \hat{\mathcal{M}}(\hat{p}(t),\hat{z}(t),t)=\hat{\mu}.\nonumber
\end{align}
By considering the derivative of $t_{mn}$ for these identities,  one obtains
\begin{align}
&\frac{\partial \mathcal{L}}{\partial p}\frac{dp}{dt_{mn}}+\frac{\partial \mathcal{L}}{\partial z}\frac{dz}{dt_{mn}}+\frac{\partial \mathcal{L}}{\partial t_{mn}}=0,\nonumber\\
&\frac{\partial \mathcal{M}}{\partial p}\frac{dp}{dt_{mn}}+\frac{\partial \mathcal{M}}{\partial z}\frac{dz}{dt_{mn}}+\frac{\partial \mathcal{M}}{\partial t_{mn}}=0,\nonumber\\
&\frac{\partial \hat{\mathcal{L}}}{\partial \hat{p}}\frac{d\hat{p}}{dt_{mn}}+\frac{\partial \hat{\mathcal{L}}}{\partial \hat{z}}\frac{d\hat{z}}{dt_{mn}}+\frac{\partial \hat{\mathcal{L}}}{\partial t_{mn}}=0,\nonumber\\
&\frac{\partial \hat{\mathcal{M}}}{\partial \hat{p}}\frac{d\hat{p}}{dt_{mn}}+\frac{\partial \hat{\mathcal{M}}}{\partial \hat{z}}\frac{d\hat{z}}{dt_{mn}}+\frac{\partial \hat{\mathcal{M}}}{\partial t_{mn}}=0.\nonumber
\end{align}
Then by applying the Lax equations of the dDS hierarchy $(\ref{LaxS})$, we derive the following Hamilton-Jacobi systems
\begin{equation}
\left(
\begin{matrix}
\dfrac{\partial \mathcal{L}}{\partial p} & \dfrac{\partial \mathcal{L}}{\partial z} \\
\dfrac{\partial \mathcal{M}}{\partial p} & \dfrac{\partial \mathcal{M}}{\partial z}
\end{matrix}
\right)
\left(
\begin{matrix}
\dfrac{dp}{dt_{mn}}\\
\dfrac{dz}{dt_{mn}}
\end{matrix}
\right)
=
\left(
\begin{matrix}
\dfrac{\partial \mathcal{L}}{\partial p} & \dfrac{\partial \mathcal{L}}{\partial z} \\
\dfrac{\partial \mathcal{M}}{\partial p} & \dfrac{\partial \mathcal{M}}{\partial z}
\end{matrix}
\right)
\left(
\begin{matrix}
\dfrac{\partial H_{mn}}{\partial z} \\
-\dfrac{\partial H_{mn}}{\partial p}
\end{matrix}
\right),\nonumber
\end{equation}
\begin{equation}
\left(
\begin{matrix}
\dfrac{\partial \hat{\mathcal{L}}}{\partial \hat{p}} & \dfrac{\partial \hat{\mathcal{L}}}{\partial \hat{z}} \\
\dfrac{\partial \hat{\mathcal{M}}}{\partial \hat{p}} & \dfrac{\partial \hat{\mathcal{M}}}{\partial \hat{z}}
\end{matrix}
\right)
\left(
\begin{matrix}
\dfrac{d\hat{p}}{dt_{mn}}\\
\dfrac{d\hat{z}}{dt_{mn}}
\end{matrix}
\right)
=
\left(
\begin{matrix}
\dfrac{\partial \hat{\mathcal{L}}}{\partial \hat{p}} & \dfrac{\partial \hat{\mathcal{L}}}{\partial \hat{z}} \\
\dfrac{\partial \hat{\mathcal{M}}}{\partial \hat{p}} & \dfrac{\partial \hat{\mathcal{M}}}{\partial \hat{z}}
\end{matrix}
\right)
\left(
\begin{matrix}
\dfrac{\partial H_{mn}}{\partial \hat{z}} \\
-\dfrac{\partial H_{mn}}{\partial \hat{p}}
\end{matrix}
\right).\nonumber
\end{equation}

\end{proof}

The trajectories of this Hamiltonian system should form a two-dimensional family because it resides in a two-dimensional phase space. Both sets of parameters $(\lambda,\mu)$ and $(\hat{\lambda},\hat{\mu})$ should have a functional relationship. Let us write the functional relations as
$$\lambda=\bar{f}(\hat{\lambda},\hat{\mu}),\ \ \ \mu=\bar{g}(\hat{\lambda},\hat{\mu}).$$
This implies that four eigenfunctions $(\mathcal{L},\mathcal{M},\hat{\mathcal{L}},\hat{\mathcal{M}})$ satisty the following equations
\begin{equation}\nonumber
\mathcal{L}(p,z,t)=\bar{f}(\hat{\mathcal{L}},\hat{\mathcal{M}}),\ \ \
\mathcal{M}(p,z,t)=\bar{g}(\hat{\mathcal{L}},\hat{\mathcal{M}}).
\end{equation}
Therefore, the Riemann-Hilbert problem can be reproduced from the multi-time Hamiltonian system.

Because there are two different multi-time Hamiltonian trajectories parameterizations, two different canonical transformations $(p,z)$ to the parameterizations $(\lambda,\mu)$ and $(\hat{\lambda},\hat{\mu})$ can be obtained. These two different canonical transformations are defined by two generation functions. One of the canonical transformations is $(p,z)\mapsto(\lambda,\mu)$, which is defined by a generating function $S(\lambda,z,\hat{z},t)$ as
\begin{align}
\frac{\partial S}{\partial \lambda}=\mu,\ \ \ & \frac{\partial S}{\partial z}=p,\nonumber\\
\frac{\partial S}{\partial \hat{z}}=\frac{u}{p},\ \ \ & \frac{\partial S}{\partial t_{mn}}=H_{mn}.\nonumber
\end{align}
The transformed Hamiltonian system with zero Hamiltonians is given by
$$\frac{d\lambda}{d t_{mn}}=0,\ \ \ \frac{d\mu}{dt_{mn}}=0.$$

Similarly, the other canonical transformation $(p,z)\mapsto(\hat{\lambda},\hat{\mu)}$ is defined by a generating function $\hat{S}=\hat{S}(\hat{\lambda},z,\hat{z},t)$ as
\begin{align}
\frac{\partial \hat{S}}{\partial \hat{\lambda}}=\hat{\mu},\ \ \ &\frac{\partial \hat{S}}{\partial z}=p,\nonumber\\
\frac{\partial \hat{S}}{\partial \hat{z}}=\frac{u}{p},\ \ \ &\frac{\partial \hat{S}}{\partial t_{mn}}=H_{mn}.\nonumber
\end{align}
The transformed Hamiltonian system with zero Hamiltonians is given by
$$\frac{d\hat{\lambda}}{d t_{mn}}=0,\ \ \ \frac{d\hat{\mu}}{d t_{mn}}=0.$$
Then the above canonical transformations can be rewritten as
\begin{align}
&dS=\mu d\lambda+pdz+\frac{u}{p}d\hat{z}+\sum_{m+n\ge2}^{\infty}H_{mn}dt_{mn},\nonumber\\
&d\hat{S}=\hat{\mu}d\hat{\lambda}+pdz+\frac{u}{p}d\hat{z}+\sum_{m+n\ge2}^{\infty}H_{mn}dt_{mn},\nonumber
\end{align}
which are actually the $S$ and $\hat{S}$ functions related to the tau functions introduced in  \cite{yi2019}.

\bigskip
\section{Conclusion}
In this paper, the existence of the tau function, the Riemann-Hilbert problem and Hamilton-Jacobi theory for the dDS hierarchy are studied. Based on these results, we will continue to consider some problems related to the twistor theory, exact solutions and algebraic structures and the results will present in the subsequent papers.

\bigskip
\bigskip
\textbf{Acknowledgements:}
This work is supported by the National Natural Science Foundation of China under Grant Nos. 12271136, 12171133 and 12171132, and the Anhui Province Natural Science Foundation No. 2008085MA05.

\end{document}